\def\year{2018}\relax
\documentclass[letterpaper]{article} 
\usepackage{aaai18}  
\usepackage{times}  
\usepackage{helvet}  
\usepackage{courier}  
\usepackage{url}  
\usepackage{graphicx}  
\usepackage{enumerate,setspace}
\usepackage{amsmath,amssymb,amsfonts}
\usepackage{amsthm}
\usepackage{graphicx}
\usepackage{booktabs}
\usepackage{fancybox}
\usepackage[english, vlined, ruled, linesnumbered]{algorithm2e}
\usepackage{fancyheadings}
\usepackage[usenames]{color}
\usepackage{tikz}
\usetikzlibrary{arrows,shapes, calc, fit, positioning}
\usepackage[small]{caption}
\usepackage{subcaption}
\newcommand{\bfx}{\boldsymbol{x}}

\newcommand{\bfw}{\boldsymbol{w}}

\newcommand{\calG}{\mathcal{G}}

\newcommand{\calM}{\mathcal{M}}

\newcommand{\calS}{\mathcal{S}}

\newcommand{\bbR}{\mathbb{R}}

\newcommand{\opt}{\mbox{\sf opt}}
\newcommand{\LP}{\mbox{\sf LP}}
\newcommand{\poly}{\mathrm{poly}}
\newcommand{\out}{\mathrm{out}}
\newcommand{\tw}{\mathrm{tw}}

\theoremstyle{plain}
	  \newtheorem{theorem}{Theorem}
	  \newtheorem{corollary}{Corollary}
	  \newtheorem{lemma}{Lemma}

\theoremstyle{definition}
	  \newtheorem{definition}{Definition}
	  
	  \newtheorem{example}{Example}
	  
\theoremstyle{remark}
	  
\newcommand*{\citet}[1]{\citeauthor{#1} \shortcite{#1}}

\begin{document}
%
\title{Cooperative Games with Bounded Dependency Degree}
\author{
Ayumi Igarashi \\ University of Oxford\\ Oxford, UK \\ ayumi.igarashi@cs.ox.ac.uk \And
Rani Izsak \\ Weizmann Institute of Science\\ Rehovot, Israel \\ ran.izsak@weizmann.ac.il \And
Edith Elkind \\ University of Oxford\\ Oxford, UK \\ edith.elkind@cs.ox.ac.uk
}
\maketitle
\begin{abstract}
Cooperative games provide a framework to study cooperation among self-interested agents.  
They offer a number of solution concepts describing how the outcome of the cooperation  
should be shared among the players. Unfortunately, computational problems 
associated with many of these solution concepts tend to be 
intractable---NP-hard or worse. In this paper, we incorporate complexity measures recently proposed by 
\citet{Feige2013}, called {\em dependency degree} and {\em supermodular degree}, 
into the complexity analysis of cooperative games. 
We show that many computational problems for cooperative games become tractable for 
games whose dependency degree or supermodular degree are bounded. In particular, we prove that simple games 
admit efficient algorithms for various solution concepts when the supermodular degree is small; further, we show 
that computing the Shapley value is always in FPT with respect to the dependency degree. Finally, we note that, 
while determining the dependency among players is computationally hard, there are efficient algorithms for 
special classes of games.
\end{abstract}

\section{Introduction}
Cooperative games provide a convenient framework to study cooperation among self-interested agents. 
Formally, a {\it cooperative transferable utility game}, or simply a {\it game}, is a pair $(N,v)$ where 
$N=\{1,2,\ldots,n\}$ is a finite set of {\it players} and $v :2^N \rightarrow \bbR$ is a \textit{characteristic 
function}. We are interested in how players should divide the value $v(N)$ of the {\em grand coalition} $N$. 
To capture the idea of a stable or fair payoff division scheme, a number of 
solution concepts have been developed, such as the core and the Shapley value. Unfortunately, the 
computational problems associated with many of these solution concepts are often intractable---NP-hard or worse.

There are two ways to circumvent computational intractability in this context. 
The first approach is to identify interesting classes of games that admit efficient algorithms. 
For instance, it is well-known that when the characteristic function $v$ is convex, 
an outcome in the core can be computed by a polynomial-time algorithm.
However, this algorithm offers no guarantees when the input game is not convex, 
even if the convexity constraint is only violated at a few points,
and hence this approach is of limited value. 
A more flexible way to cope with hardness is to 
provide complexity guarantees for 
all instances so that the 
guarantee depends on the complexity of the instance, that is, to design algorithms whose
running time depends on how well the input instance is structured. 
In the context of cooperative games, this approach has been pursued by \citet{Ieong2005},
who propose a representation formalism, which they call {\em marginal contribution nets (MC-nets)}, 
and design an algorithm for computing an allocation in the core
whose running time depends on the treewidth of the graph associated with the MC-net representation
of the input game.

In this work, we explore the power of the latter approach 
for two measures of structural complexity of set functions
that have been recently developed by \citet{Feige2013}: 
the {\em dependency degree} and the {\em supermodular degree}. 
Intuitively, the complexity of a set function $v$ is measured 
using the notion of {\em dependency} among players at $v$; 
such dependencies induce a graph describing the relation between players---the {\em dependency graph}. 
A player's dependency degree is her degree in this graph; 
her supermodular degree is her degree in a modified version of this graph,
which only takes certain dependencies into account.

\smallskip

\noindent {\bf Our contribution.\ }
We argue that both the dependency degree and the supermodular degree
are useful in the context of cooperative games, both analytically and computationally. 

We show that several cooperative game theory concepts can be naturally interpreted in terms of a dependency 
graph. In particular, a players' dependency degree reflects on her importance in the game: 
in a simple game, a dummy player is an isolated vertex of the dependency graph and 
a veto player is a vertex with the maximum degree. 
We can also relate properties of a game, such as its dimension, 
to the properties of its dependency graph: 
for instance, we show that dependency graphs of weighted voting games 
are clique-trees, namely, chordal graphs.

We then investigate which solution concepts in 
cooperative games can be computed efficiently if the dependency/supermodular degree is bounded. 
For simple games, we obtain a number of tractability results
with respect to the supermodular degree. Specifically, we prove that simple 
games admit an efficient algorithm for computing an element of the core or the least core when 
the supermodular degree is small. Further, while finding an optimal coalition structure is computationally 
intractable even for weighted voting games, we show that this problem becomes tractable for weighted voting 
games with small supermodular degree.
However, these results do not extend to general games: we prove that the associated separation problem for the least core is 
NP-hard even for games with constant dependency degree. On the other hand, we show that computing the Shapley value 
and the Banzhaf value is in FPT with respect to the dependency degree.

We also consider the problem of computing the dependency degree and the supermodular degree 
given various representations of a game. While intractability turns out to be inevitable 
in general, we provide polynomial and pseudo-polynomial algorithms for special
classes of games.

\smallskip

\noindent {\bf Related work.\ }
Computational aspects of cooperative games have received a considerable amount of
attention over the last few decades; we refer the reader to the book of \citet{Chalkiadakis2011}.

Our work is similar in spirit to the complexity study of induced subgraph games or, more broadly, 
games defined by MC-nets \cite{Deng1994,Ieong2005,Greco2011,Greco2014,LiCo14}.
Indeed, each MC-net induces an {\em agent graph}, which also aims
to capture dependencies among agents. However agent graphs are defined in a purely syntactic manner, 
by looking at agents' co-occurrences in the rules of an MC-net, whereas the dependency graph 
is defined semantically, i.e., in terms of the value of the characteristic function.
Moreover, both the dependency degree and the supermodular degree are different from the concepts
that are usually used to measure the complexity of an agent graph (such as treewidth).

There are also similarities between our model and {\em Myerson games} \cite{Myerson1977,Chalkiadakis2016,Meir2013,Igarashi2017} 
where players are located on a graph and coalitions are only allowed to form if they are connected in this graph. 
However, in Myerson games non-adjacent agents may still depend on each other, and hence an agent's dependency degree
may be high even if her degree in the underlying Myerson graph is small. As a consequence, 
some problems that are easy for games with small supermodular degree remain hard for 
games on bounded-degree graphs, even if these graphs are acyclic \cite{Igarashi2017}.
Our results for coalition structure generation (Section~\ref{sec:csg}) are similar in spirit to those
of \citet{Voice2012}; we discuss the relationship between their results and ours
in Section~\ref{sec:csg}.

The dependency degree and the supermodular degree have been introduced by~\citet{Feige2013},
who showed applications of these measures to the welfare maximization problem.
\citet{Feldman2014} generalized these results to function maximization subject to $k$-extendible system constraints 
(a generalization of the intersection of $k$ matroids). These concepts have also been
applied in an online setting \cite{Feldman2017}, and in the context of efficiency of auctions \cite{FFMR16}, 
optimization of SDN upgrades \cite{Poularakis2017} and committee selection \cite{Izsak2017}.
Some related complexity measures are the submodularity ratio \cite{DasKempe} and MPH~\cite{MPH}.

\begin{table*}[ht]
	\centering
	\begin{tabular}{lcccc}
		\toprule
		& General games & Simple games & Weighted voting games \\
	     	\midrule
Optimal coalition structure & NP-h. & NP-h. for $p=6$ (Th. \ref{NPh:optcoalition}) & FPT wrt p (Cor. \ref{FPT:optcoalition}) \\
	     	\midrule
Core &  & P & P\\
          	\midrule
Least Core &  & FPT wrt $p$ (Th. \ref{FPT:LC}) & FPT wrt $p$\\
          	\midrule
Shapley &  FPT wrt $d$  (Th. \ref{FPT:Shapley}) &  FPT wrt $d$  &  FPT wrt $d$\\
          	\midrule
Banzhaf & FPT wrt $d$  (Th. \ref{FPT:Banzhaf}) &  FPT wrt $d$ &  FPT wrt $d$\\
		\bottomrule
	\end{tabular}
	\vspace{3pt}
	\caption{Overview of complexity results for computing various solution concepts 
		when parameterized by the dependency degree $d$ and the supermodular degree $p$. 
		Note that deciding the non-emptiness of the core for a simple game is straightforward: 
		the core is non-empty if and only if there is a veto player \cite{Chalkiadakis2011}.}
	\vspace{-5pt}
	\label{table}
\end{table*}

\section{Preliminaries}
We start by defining basic notation and terminology of cooperative games.
Recall that a cooperative game is a pair $(N, v)$, where $N$ is a finite set
and $v$ is a function from $2^N$ to $\mathbb R$. Throughout the paper, we assume $v(\emptyset)=0$. 
For $s\in{\mathbb N}$, let $[s]=\{1,2,\ldots,s\}$. 
For a vector $\bfx \in \bbR^n$ and a subset $S\subseteq N$ 
we use the notation $x(S)= \sum_{i \in S}x_i$, with the convention that 
$x(\emptyset)=0$. The subsets of $N$ are referred to as {\em coalitions}.
An {\it imputation} for a game $(N,v)$ is a vector $\bfx \in \bbR^N$ 
satisfying {\em efficiency} : $x(N)=v(N)$, and {\em individual rationality} : $x_i \geq v(\{i\})$ for every $i \in N$.
For a player $i \in N$ and a coalition $S \subseteq N \setminus \{i\}$, 
we let $v(i|S)=v(S \cup \{i\})-v(S)$, that is,
$v(i|S)$ is the {\em marginal contribution} of $i$ to $S$. 
A player $i$ is called a {\em dummy} if she does not contribute to any coalition, 
i.e., $v(i|S)=0$ for every $S \subseteq N$. 

A set function $v:2^N\to{\mathbb R}$ is said to be {\em monotone} if for every pair of subsets
$S, T\subseteq N$ it holds that $S \subseteq T$ implies $v(S)\leq v(T)$.
A game $(N,v)$ is said to be {\em simple} if 
$v$ is monotone and only takes values in $\{0, 1\}$. 
In a simple game, coalitions of value $1$ are said to be {\em winning}, 
and coalitions of value $0$ are said to be {\em losing}. 
A winning coalition $S$ is said to be {\em minimal} if removal of any player from 
$S$ makes it losing, i.e., $S \setminus \{i\}$ is losing for every $i \in S$. 
A player $i$ is said to be a {\em veto player} if she is present in all winning coalitions, i.e., 
if a coalition $S$ is winning, then $i \in S$. A player $i$ is a {\em pivot} for a coalition 
$S \subseteq N \setminus \{i\}$ if $S$ is losing and $S \cup \{i\}$ is winning.

The {\em core} is a classic solution concept in cooperative games. 
Formally, the core of a game $(N,v)$ is the set of all imputations $\bfx$ 
such that no coalition has an incentive to defect from $\bfx$, i.e., 
$x(S) \geq v(S),~\mbox{for all}~S \subseteq N$. As the core can be empty, 
and, on the other hand, not all outcomes in the core are equally fair, 
we consider the {\em least core}, which can be thought of as the set of most stable outcomes. 
We first define the {\em excess} of a coalition $S \in 2^N \setminus \{N, \emptyset\}$ 
at an imputation $\bfx$ as $e(\bfx,S):=v(S)-x(S)$; intuitively, $e(\bfx, S)$ 
is the degree of unhappiness of $S$ at $\bfx$.
The {\em least core} of a game $(N,v)$ is the set of all imputations $\bfx$ 
that minimize the maximum excess, i.e., the set of optimal solutions $\bfx$ 
of the following linear program: 
\begin{alignat*}{5}
(\LP_0)~ \min        & \quad  \varepsilon &\\
 \text{s.t.} & \quad  x(S) \geq v(S) - \varepsilon \quad \mbox{for all}~S \in 2^N \setminus \{N, \emptyset\}\\
& \quad x_i \geq v(\{i\}) \quad \mbox{for all}~i \in N\\
& \quad x(N)=v(N).
\end{alignat*}

We also consider solution concepts capturing fairness among players: the Shapley value and the Banzhaf value. 
The {\em Shapley value} $\phi_i(N,v)$ of a player $i\in N$ in a game $(N,v)$ 
is the average of  $i$'s marginal contributions at $v$ over all permutations of the players, that is,
\[
\phi_i(N,v)= \sum_{S \subseteq N \setminus \{i\}}\frac{|S|! (n-|S| -1)!}{n!}v(i|S).
\]
The {\em Banzhaf value} is the average of $i$'s marginal contributions at $v$ over all coalitions, that is, 
\[
\beta_i(N,v)=\frac{1}{2^{n-1}} \sum_{S \subseteq N \setminus \{i\}} v( i | S).
\]

\noindent {\bf Computational setting}
Throughout the paper, we only consider games $(N,v)$ such that $v$ is computable in time 
polynomial in $n$. Furthermore, `polynomial' always means polynomial in the number 
of players $n$. Note that the explicit representation of a game $(N, v)$, which lists the values
of all coalitions, is not polynomial in $n$; thus, for our computational results for general games
we assume oracle access to the characteristic function $v$.
We say that a problem is {\em fixed parameter tractable} (FPT) with respect to a 
parameter $k$ if each instance $I$ of this problem can be solved in time $f(k) \poly(|I|)$, 
where $f$ is a computable function that depends on $k$ only.

We omit some proofs due to space constraints; the omitted proofs
can be found in the full version of the paper \cite{IIE17}.

\section{Dependency Graphs of Cooperative Games}
In this section, we introduce complexity measures representing dependencies among players in a cooperative game 
\cite{Feige2013}, and study how well such parameters capture important concepts in cooperative games.

Given a game $(N, v)$ and two players $i, j\in N$, 
we say that player $i$ {\em positively depends on} player $j$ if there exists a coalition 
$S \subseteq N \setminus \{i,j\}$ such that $v(i|S \cup \{j\}) > v(i|S)$,
i.e., $i$ can contribute more to $S$ in the presence of $j$. 
We say that $i$ {\em depends on} $j$ if there exists a coalition $S \subseteq N \setminus \{i,j\}$ 
such that $v(i|S \cup \{j\}) \neq v(i|S)$, i.e., $i$'s contribution to $S$ depends on the presence of $j$.
These relations are known to be symmetric \cite{Feige2013}; hence, we can 
capture the dependency relations between players by undirected graphs. 
Formally, we define the {\em supermodular dependency graph} $G^+_v$ 
to be an undirected graph where the set of nodes is given by $N$ 
and the set of edges is given by the pairs of players positively depending on each other.
The {\em supermodular dependency set} of $i$ at $v$ is 
\[
D^{+}(i)=\{\, j \in N\setminus \{i\} \mid \mbox{$i$ positively depends on $j$} \,\}.
\]
The {\em supermodular degree} $p$ of $v:2^N \rightarrow \bbR$ 
is defined as the maximum size of the supermodular dependency set, i.e., $p=\max_{i \in N} |D^+(i)|$. 
We define the {\em dependency graph} $G_v$ to be an undirected graph 
where the set of nodes is given by $N$ and the set of edges is given by the pairs of players depending on each other. 
The {\em dependency set} of $i$ at $v$ is
\[
D(i)=\{\, j \in N \setminus \{i\}\mid  \mbox{$i$ depends on $j$}\,\}.
\]
The {\em dependency degree} $d$ of $v:2^N \rightarrow \bbR$ is defined as the maximum size of the dependency 
set, i.e., $d=\max_{i \in N} |D(i)|$.

\begin{example}\label{ex:dependency}
Consider a simple game $(N,v)$ with player set $N=[4]$ and a characteristic function $v:2^N \rightarrow \{0,1\}$ 
where the set of minimal winning coalitions is
$\{\{1,2\},\{2,3\},\{3,4\},\{1,4\}\}$.
It is easy to see that only such pairs have a positive dependence. Further, player $1$ depends on player $3$ 
since $1$ can make a positive marginal contribution to the coalition $\{2,4\}$,
but her contribution becomes zero in the presence of $3$. A similar argument applies to the pair $2, 4$. 
The resulting dependency graphs $G^+_v$ and $G_v$ are depicted in Figure~\ref{fig:dependency}. 
\begin{figure}[htb]
\begin{subfigure}[t]{0.25\textwidth}
\centering
\begin{tikzpicture}[scale=0.8, transform shape]
		\node[draw, circle](1) at (0,0) {$1$};
		\node[draw, circle](2) at (1.2,0) {$2$};
		\node[draw, circle](3) at (1.2,-1.2) {$3$};
		\node[draw, circle](4) at (0,-1.2) {$4$};
		
		\draw[-, >=latex,thick] (1)--(2); 
		\draw[-, >=latex,thick] (2)--(3); 
		\draw[-, >=latex,thick] (3)--(4);
		\draw[-, >=latex,thick] (1)--(4);
\end{tikzpicture}
	\caption{The graph $G^+_v$}
	\label{fig1}
\end{subfigure}%
\begin{subfigure}[t]{0.25\textwidth}
\centering
\begin{tikzpicture}[scale=0.8,transform shape]
		\node[draw, circle](1) at (0,0) {$1$};
		\node[draw, circle](2) at (1.2,0) {$2$};
		\node[draw, circle](3) at (1.2,-1.2) {$3$};
		\node[draw, circle](4) at (0,-1.2) {$4$};
		
		\draw[-, >=latex,thick] (1)--(2); 
		\draw[-, >=latex,thick] (2)--(3); 
		\draw[-, >=latex,thick] (3)--(4);
		\draw[-, >=latex,thick] (1)--(4);
		\draw[-, >=latex,thick] (2)--(4); 
		\draw[-, >=latex,thick] (1)--(3); 
\end{tikzpicture}
\caption{The graph $G_v$}
\label{fig2}
\end{subfigure}
\caption{A supermodular dependency graph and a dependency graph (Example~\ref{ex:dependency})}
\label{fig:dependency}
\end{figure}
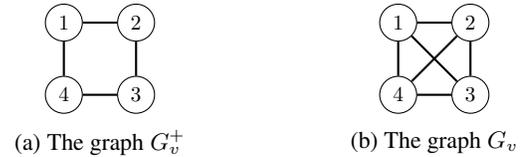
\end{example}

We will now show that the parameters defined above capture the importance of a player in the underlying game. 
The theorem below shows that a dummy player corresponds to an isolated node in the dependency graph.

\begin{theorem}
For every cooperative game $(N, v)$ and every player $i\in N$, the following statements are equivalent:
\begin{itemize}
\item[$(${\rm i}$)$] Player $i$ is a dummy player.
\item[$(${\rm ii}$)$] $v(\{i\})=0$ and $D(i)=\emptyset$.
\item[$(${\rm iii}$)$] $v(\{i\})=0$ and $D^{+}(i)=\emptyset$.
\end{itemize}
\end{theorem}
\begin{proof}
$(${\rm i}$)\Longrightarrow(${\rm ii}$)$: 
Suppose that $i$ is a dummy player. Then we have $v(\{i\})=0$. Now, suppose that there exists a 
player $j \in D(i)$. Then there is a coalition $S$ such that at least one of $v(i|S)$ and $v(i|S \cup \{j\})$ is 
non-zero, contradicting the fact that $i$ is a dummy player.

\noindent
$(${\rm ii}$)\Longrightarrow(${\rm iii}$)$: This direction holds by the definition.

\noindent
$(${\rm iii}$)\Longrightarrow(${\rm i}$)$: 
Suppose that $v(\{i\})=0$ and $D^{+}(i)=\emptyset$. 
Assume towards a contradiction that $v(i|S) \neq 0$ for some coalition $S \subseteq N \setminus \{i\}$.  
Let $S^*$ be a minimal coalition with respect to this property. If $S^* \neq \emptyset$, then there is a player 
$j \in S^*$ such that $v(i| S^* \setminus \{j\}) =0$, which means that $i$ can contribute more in the 
presence of $j$, and hence $i$ positively depends on $j$, a contradiction. If $S^* = \emptyset$, it follows 
that $v(\{i\}) \neq 0$, a contradiction again.
\end{proof}


\subsection{Positive Dependence in Simple Games}
We will now further investigate the structure of dependency graphs in simple games. It turns out that such 
games can be almost fully characterized by the positive dependence relation. We first observe that in simple 
games the dependency relation admits a natural interpretation.
\begin{lemma}\label{lem:simple:positive}
Consider a simple game $(N, v)$ and two players $i, j\in N$.
Player $i$ positively depends on player $j$ if and only if 
there exists a coalition $S \subseteq N \setminus \{i,j\}$ such that 
$S \cup \{i,j\}$ is winning, but the coalitions $S \cup \{i\}$ and $S \cup \{j\}$ are losing.
\end{lemma}

Now, in contrast with dummy players, a veto player in a simple game 
is adjacent to every non-dummy player in the dependency graph.

\begin{theorem}
Let $i$ be a veto player in a simple game $(N, v)$. 
Then $i$ positively depends on each non-dummy player $j \in N \setminus \{i\}$.
\end{theorem}
\begin{proof}
Take any non-dummy player $j \in N \setminus \{i\}$. Then $j$ is a pivot for some coalition 
$S \subseteq N \setminus \{j\}$, i.e., $S \cup \{j\}$ is winning and $S$ is losing. 
Since $i$ is a veto player and $i \neq j$, we have $i \in S$. 
Let $S_{-i}=S \setminus \{i\}$. Then, the coalition $S_{-i} \cup \{i,j\}$ 
is winning, and coalition $S_{-i}\cup \{i\}$ is losing. 
Also, as $i \not \in S_{-i}\cup \{j\}$, coalition $S_{-i}\cup \{j\}$ is losing. 
By Lemma~\ref{lem:simple:positive}, we conclude that $i$ positively 
depends on $j$.
\end{proof}

We also observe that all minimal winning coalitions correspond to cliques in the supermodular dependency graph.

\begin{theorem}\label{thm:MWC}
In a simple game $(N, v)$ player $i\in N$ positively depends on 
a player $j\in N$ if and only if there exists a minimal winning coalition $S \subseteq N$ such that $i,j \in S$.
\end{theorem}
\begin{proof}
Suppose that $i$ positively depends on $j$. By Lemma~\ref{lem:simple:positive}
there exists a coalition $S \subseteq N \setminus \{i\}$ such that $S \cup \{i,j\}$ is winning, 
but coalitions $S \cup \{j\}$ and $S \cup \{i\}$ are losing. 
Choose a minimal coalition $S^*$ with respect to this property. Then, removing any 
player from $S^* \cup \{i,j\}$ makes this coalition losing, and hence $S^* \cup \{i,j\}$ is a minimal winning 
coalition. Conversely, suppose that $i,j$ belong to some minimal winning coalition $S$. Since $S$ is a minimal 
winning coalition, both $S\setminus\{j\}$ and $S\setminus\{i\}$ are losing, but $S$ is winning. 
By Lemma~\ref{lem:simple:positive}, $i$ positively depends on $j$.
\end{proof}

\begin{corollary}\label{cor:MWC}
Consider a simple game $(N,v)$ and a minimal winning coalition $S \subseteq N$.
Any two distinct players $i, j\in S$ positively depend on each other.
\end{corollary}

{\em Weighted voting games} form a subclass of simple games. Such games can be 
succinctly represented by the {\em weight} of each player, and a {\em quota}.
Formally, a weighted voting game with a set of players $N=[n]$ is given by a list of non-negative weights 
$\bfw=(w_1,w_2,\ldots,w_n)$ and a quota $q$; we will write $[N;\bfw;q]$. Its characteristic function 
$v:2^N \rightarrow \{0,1\}$ is given by
\[
v(S)=
\begin{cases}
&1~\mbox{if}~\sum_{i \in S}w_i \geq q\\
&0~\mbox{otherwise}.
\end{cases}
\]
It turns out that the supermodular dependency graph of a weighted voting game has a special structure: 
we will show that any such graph is a {\em chordal graph}.
Recall that a graph is said to be {\em chordal} if any cycle of four or more nodes has a {\em chord}, 
i.e., an edge that does not belong to the cycle, but connects two of its nodes.

We first state and prove the following lemma. 

\begin{lemma}\label{lem:chordal:supermodular}
Consider a weighted voting game $[N;\bfw;q]$ where a player $i\in N$ 
positively depends on a player $k\in N$ and 
$w_i \geq w_j \geq w_k$. Then, player $i$ positively depends on player $j$.
\end{lemma}
\begin{proof}
Since $i$ positively depends on $k$, there is a coalition $S$ such that $S \cup \{i,k\}$ is winning, but $S \cup \{i\}$ and $S \cup \{k\}$ are losing.
Now, since $w_{i} \geq w_{j} \geq w_{k}$, $S \cup \{i,j\}$ is winning, but $S \cup \{i\}$ and $S \cup \{j\}$ are losing, implying that $i$ positively depends on $j$. 
\end{proof}

\begin{theorem}\label{thm:chordal:supermodular}
For every weighted voting game $[N;\bfw;q]$ it holds that its supermodular dependency graph is chordal.
\end{theorem}
\begin{proof}
Suppose towards a contradiction that $G^{+}_v$ has a chordless cycle $C$ of at least four players. Let 
$C=\{i_1,i_2,\ldots,i_k\}$, where $i_j$ positively depends on $i_{j+1}$ for $j \in [k]$ 
(with the convention that $i_{k+1}=i_1$). 
Assume without loss of generality that $i_1$ has the maximum weight among the players in $C$ and 
$w_{i_2} \leq w_{i_k}$.

First, suppose that there is a player $i_j \in C \setminus \{i_1,i_2,i_k\}$ such that 
$w_{i_j} \geq w_{i_2}$. Since $i_1$ positively depends on $i_2$ and 
$w_{i_1} \geq w_{i_j} \geq w_{i_2}$, by Lemma \ref{lem:chordal:supermodular}
this means that $i_1$ positively depends on $i_j$, a contradiction.
Now suppose that for all players $i_j \in C \setminus \{i_1,i_2,i_k\}$ we have $w_{i_j} <w_{i_2}$; 
in particular, we have $w_{i_{k-1}} <w_{i_2}$. Since
$i_{k}$ positively depends on $i_{k-1}$, and $w_{i_k} \geq w_{i_2} \geq w_{i_{k-1}}$, 
by Lemma \ref{lem:chordal:supermodular} player $i_k$ positively depends 
on $i_2$, a contradiction.
\end{proof}

A simple game $(N,v)$ is the {\em intersection} of $k$ weighted voting games 
$[N;\bfw^{\ell};q^{\ell}]$, $\ell \in [k]$, if for every coalition $S \subseteq N$
we have  $v(S)=1$ if and only if $w^{\ell}(S) \geq q^{\ell}$ for all $\ell \in [k]$. 
It is known that every simple game $G$ can be represented as an intersection of multiple weighted voting games;
the minimum number of weighted voting games whose intersection equals to $G$
is called the {\em dimension} of~$G$.
 
Observe that the game defined in Example~\ref{ex:dependency} has dimension $2$: 
it can be represented as the intersection of weighted voting games $[N;\bfw^{1};1]$ and $[N;\bfw^{2};1]$, 
where $w^{1}(1)=w^{1}(3)=w^{2}(2)=w^{2}(4)=1$, and all other weights are zero. However, 
its supermodular dependency graph has a chordless cycle of length four; 
thus, we cannot guarantee that the supermodular dependency graph of a simple 
game is chordal beyond dimension $1$.



\section{Complexity of Stability-Related Solution Concepts}
In this section, we investigate the complexity of computing outcomes 
in the core and the least core. 

For simple games, it is well-known that deciding if the core
is not empty or finding an outcome in the core is easy. We now complement this result
by showing that computing an element of the least core in simple games
is fixed-parameter tractable with respect to the supermodular degree. 

\begin{theorem}\label{FPT:LC}
Let $(N,v)$ be a simple game. Given its supermodular dependency graph $G^+_v$
and oracle access to $v$, we can compute an element of the least core in FPT 
time with respect to the supermodular degree.
\end{theorem}
\begin{proof}
We first check whether the input game admits an imputation, i.e., whether $v(N) \geq \sum_{i \in N}v(\{i\})$: 
if not, the least core is empty. Thus, from now on we assume that $v(N) \geq \sum_{i \in N}v(\{i\})$.

It suffices to show that the separation problem for the linear program $\LP_0$ is fixed-parameter 
tractable with respect to the supermodular degree. Fix an $\varepsilon \in \bbR$ and $\bfx \in \bbR^N$. 
First, we can clearly check in polynomial time
whether $\bfx$ is an imputation; thus, in the rest of the proof we assume that this is indeed the case.
We need to show that 
deciding whether the following inequality holds is in FPT with respect to $p$:
\begin{equation}\label{eq:1}
\min \{\, x(S)-v(S) \mid S \in 2^N \setminus \{N,\emptyset\} \,\} \geq \varepsilon.
\end{equation}
Since $\bfx$ is an imputation, we have $x(N) \leq 1$ and $x_i \ge 0$ for all $i \in N$. 
Now, by non-negativity of $\bfx$ and by the fact that $x(S)\leq 1$ for all $S \subseteq N$, 
the term on the left can be rewritten as 
\begin{align*}
&\min \{\, x(S)-v(S) \mid S \in 2^N \setminus \{N,\emptyset\} \,\}\\
&= \min \{\, x(S) -1  \mid S \in \calM \, \},
\end{align*}
where $\calM$ is the set of all minimal winning coalitions in $2^N \setminus \{N,\emptyset\}$. 
It remains to show that the computation of $\min \{\, x(S) -1  \mid S \in \calM \, \}$ 
is in FPT with respect to the supermodular degree. Fix $i \in N$ and let $\calM(i)$ denote 
the set of all coalitions in $\calM$ including $i$. By Corollary~\ref{cor:MWC}, 
these coalitions are subsets of $D^{+}(i)\cup \{i\}$, i.e., 
$\calM(i) \subseteq 2^{D^{+}(i) \cup \{i\}}$. By iterating through all $i \in N$ 
and all subsets of $D^{+}(i)\cup \{i\}$, we can compute the value 
$\min \{\, x(S) -1  \mid S \in \calM \, \}$. 
\end{proof}

When the dependency graph has degree at most $2$, \citet{Feige2013} showed that demand queries can be answered in 
polynomial time; that is, given a characteristic function $v:2^N \rightarrow \bbR$ and a vector $\bfx \in 
\bbR^N$, one can efficiently compute a subset $S^*$ maximizing $v(S^*)-x(S^*)$ over any subfamily of $2^N$. This 
allows us to obtain the following result.

\begin{theorem}\label{poly:LC:Nucleolus}
Consider a game $(N,v)$ whose dependency degree is at most $2$. Given oracle access
to the characteristic function $v$, we can decide the 
non-emptiness of the core or find an element of the least core
in time polynomial in $n$. 
\end{theorem}
\begin{proof}
We will argue that the separation problem for $\LP_0$ 
can be solved in time polynomial in $n$. Consider an $\varepsilon \in \bbR$ and 
a vector $\bfx \in \bbR^N$. Again, one can decide whether 
$\bfx$ is an imputation in time $O(n)$. Now, if the dependency degree is at most $2$, 
it follows from the work of \citet{Feige2013} that the maximum value $v(S)-x(S)$ over 
coalitions in $S \in 2^{N} \setminus \{N,\emptyset\}$ can be computed in 
time polynomial in~$n$; thus, we can compare the maximum with the given $\varepsilon$ and 
efficiently decide whether $(\varepsilon,\bfx)$ satisfies the inequalities in~$\LP_0$.

It remains to notice that the imputations in the core are exactly the optimal solutions to $\LP_0$ with $\varepsilon$ 
replaced with the value $0$, i.e., we can use the same procedure as above to decide whether the core is non-empty.
\end{proof}

However, this argument does not extend to general cooperative games:
we will now demonstrate that there is a succinctly representable class of games 
for which the separation problem for $\LP_0$ is NP-hard,
even though the dependency degree of games in this class is bounded
by a small constant and their supermodular degree is $1$,

\begin{theorem}
There exists a family of hypergraph games with dependency degree~7 and supermodular degree~1
for which the separation problem for $\LP_0$ is {\em NP}-hard. 
\end{theorem}
\begin{proof}
An instance of {\sc 3-Regular Independent Set} is given by a 3-regular graph $G$
and an integer $k$; it is a `yes'-instance if $G$ has an independent set of size at least $k$
and a `no'-instance otherwise. This problem is known to be NP-hard \cite{garey}.
We will now show how to reduce it to the separation problem for $\LP_0$ for a family of hypergraph games defined below.
For every vertex $v \in V$, we introduce two players $v^1, v^2$, and let 
$N =\{v^1, v^2\mid v\in V\}\cup\{d^1, d^2\}$. 
We define a set function $v$ by building its hypergraph representation $H=(N, E_H, w)$, 
so that for each $S\subseteq N$ the value 
$v(S)$ is computed as the sum of the weights of all hyperedges of the sub-hypergraph induced by $S$.
We set $w(u)=0$ for each $u\in N$. For each $v\in V$ we connect $v_1$ and $v_2$ 
by an edge (i.e., a hyperedge of rank~2) of weight~3, and we connect $d^1$ and $d^2$
by an edge of weight $|V|/2$.
For every edge $e=\{u,v\} \in E$, we introduce a hyperedge of size~4 and weight $-1$
containing $u^1$, $u^2$, $v^1$ and $v^2$. Since $G$ is 3-regular, $|E|=\frac{3}{2}|V|$ and hence 
$v(N)=3|V|-\frac{3}{2}|V|+\frac{1}{2}|V|=2|V|$.
Moreover, the constructed game has dependency degree $d=7$ since $G$ is 3-regular, 
and there is a pair of dependent vertices 
for every vertex in $G$. Also, the supermodular degree $p$ of this game is $1$, 
as the supermodular dependency set of every vertex only contains 
the other vertex belonging to the same pair.

Now, we set $x^*_u=1$ for each $u\in N\setminus\{d^1, d^2\}$ and $x^*_{d^1}=x^*_{d^2}=0$.
Clearly, the vector $\bfx^*$ is an imputation, as we have $v(N)=x^*(N)=2|V|$ and the constraint $x^*_i \geq v(\{i\})$ is satisfied for every $i \in N$.

It can be shown that $G$ admits an independent set of size $k$ if and only if the maximum excess at $\bfx^*$ is at least $k+\frac{|V|}{2}$.
Indeed if $I\subseteq V$ is an independent set of size $k$ in $G$, then $\{v^1, v^2: v\in I\}\cup\{d^1, d^2\}$ is a coalition whose excess at $\bfx^*$ 
is $k+|V|/2$. Conversely, it can be argued that if there is a coalition whose excess at $\bfx^*$ is $k + \frac{|V |}{2}$ then $G$ has an independent set of size $k$. 
\end{proof}



\section{Complexity of the Shapley and Banzhaf Values}
For the Shapley and Banzhaf values, the following observation is crucial for our analysis:
for every player $i$ and every coalition $C\subseteq N\setminus\{i\}$ it holds that
adding players who do not depend on $i$ to $C$
does not affect $i$'s marginal contribution to $C$. 
We formalize this observation in the following lemma. We write 
$I(i)=N\setminus (D(i)\cup \{i\})$.

\begin{lemma}\label{lem:shapley}
For every $S\subseteq D(j)$ and for every $T \subseteq I(j)$ it holds that
$v(j|S\cup T)=v(j|S)$.
\end{lemma}
\begin{proof}
We prove our claim by induction on the size of the set $T$.
The claim clearly holds when $|T|=1$. Suppose it holds for $|T| \leq k-1$. Let $T=\{i_1,i_2,\ldots,i_k\}$. Then
$v(j|S\cup \{i_1,i_2,\ldots,i_{k-1}\})=v(j|S)$ by the induction hypothesis and 
$v(j|S\cup \{i_1,i_2,\ldots,i_{k-1},i_k\})=v(j|S\cup \{i_1,i_2,\ldots,i_{k-1}\})$ by the fact that $i_k \in I(j)$.
Combining these two equations yields
\[
v(j|S\cup T)=v(j|S\cup \{i_1,i_2,\ldots,i_{k-1},i_k\})=v(j|S).
\]
\end{proof}

By Lemma \ref{lem:shapley}, one can calculate the Shapley value of each player by iterating over all subsets of 
her dependency set, computing her marginal contribution, and counting the number of coalitions whose intersection with 
the dependency set is exactly this subset.

\begin{theorem}\label{FPT:Shapley}
Computing the Shapley value is in FPT with respect to the dependency degree.
\end{theorem}
\begin{proof}
The Shapley value can be written as follows:
\begin{align*}
&\phi_i(N,v)=\sum_{S \subseteq N \setminus \{i\}}\frac{|S|! (n-|S| -1)!}{n!}v(i|S),\\
&=\frac{1}{n!} \sum_{S \subseteq D(i)} \sum_{T \subseteq I(i)} \frac{|S\cup T|! (n-|S \cup T| -1)!}{n!} ~v(i| S \cup T),\\
&=\frac{1}{n!} \sum_{S \subseteq D(i)} v(i| S) \sum^{n-|S|-1}_{t=0} \alpha(S,t),
\end{align*}
where $\alpha(S,t)=\binom{n-|S|-1}{t} \frac{(|S|+t)! (n-|S|-t-1)!}{n!}$. The last equality holds due to 
Lemma~\ref{lem:shapley}.
\end{proof}

Similarly, the Banzhaf value can be efficiently computed when the dependency degree is bounded. 

\begin{theorem}\label{FPT:Banzhaf}
Computing the Banzhaf value is in FPT with respect to the dependency degree.
\end{theorem}
\begin{proof}
The Banzhaf value can be written as follows:
\begin{align*}
\beta_i(N,v)&= \frac{1}{2^{n-1}} \sum_{S \subseteq N \setminus \{i\}} v( i | S),\\
&=\frac{1}{2^{n-1}} \sum_{S \subseteq D(i)} \sum_{T \subseteq I(i)}v(i| S \cup T),\\
&=\frac{1}{2^{n-1}} \sum_{S \subseteq D(i)} v(i|S) ~ 2^{(n-|S|-1)},
\end{align*}
where the last equality holds due to Lemma~\ref{lem:shapley}.
\end{proof}


\section{Optimal Coalition Structure Generation}\label{sec:csg}
So far, we discussed games {\em without} coalition structures: we implicitly assume that all players cooperate and 
are willing to divide the value of the grand coalition.
In some settings, however, it may be more efficient to split the players into different teams. The problem of 
finding the best partition of players, which is referred to as the {\em optimal coalition structure generation 
problem}, has thus been extensively studied (see, e.g., the surveys \cite{rahwan,Rahwan2015}).
Formally, a {\em coalition structure} for $N$ 
is a partition $\pi=\{S_1,S_2,\ldots,S_{\ell}\}$ of $N$ into disjoint coalitions.
A coalition structure $\pi$ for $N$ is said to be {\em optimal} 
if the {\em social welfare} $\sum_{S \in \pi} v(S)$ is maximized.

Not surprisingly, optimal coalition structure generation is NP-hard even for weighted voting games; 
this can be shown by a straightforward reduction from {\sc Partition} \cite{Chalkiadakis2011}. 
In contrast, we will now argue that if the input game is simple
and its supermodular dependency graph has tree-like structure, this problem 
becomes tractable.
To this end, we introduce the notions of {\em tree decomposition} and {\em treewidth}. 
\begin{definition}\label{def:treedecomposition}
A {\em tree decomposition} of a graph $G$ is a pair $(T,(X_{t})_{t \in V(T)})$, 
where $T$ is a rooted tree and $(X_{t})_{t \in V(T)}$ is a family of subsets of $V(G)$, called {\it bags}, where
\begin{enumerate}
\item[(i)] for every node $i \in V(G)$, the set $X^{-1}(i):=\{\, t \in V(T) \mid i \in X_{t}\,\}$ is nonempty and connected in $T$, and
\item[(ii)] for every edge $\{i,j\} \in E(G)$, there is a node $t \in V(T)$ such that $i,j \in X_{t}$.
\end{enumerate}
\end{definition}
For a node $t$ of $T$, we let $V_t$ be the union of all bags present in the subtree of $T$ rooted at $t$, including $X_t$.
The {\em treewidth} of a tree decomposition $(T,(X_{t})_{t \in V(T)})$ of $G$ is 
$\max_{t \in V(T)}(|X_t|-1)$. 

Recall that for simple games it holds that 
minimal winning coalitions form cliques in the supermodular dependency graph. 
We will now argue that this implies 
that there is an optimal coalition structure where each winning coalition 
is contained in some bag of the tree decomposition. 

\begin{lemma}\label{lem:treedecomposition}
Let $(N,v)$ be a simple game and let $(T,(X_{t})_{t \in V(T)})$ be a tree decomposition 
of the supermodular dependency graph $G^+_v$. For every $t \in V(T)$ it holds that
every optimal coalition structure $\pi$ for $V_{t}$ can be transformed 
into another optimal coalition structure $\pi^{\prime}$ 
so that for each coalition $S\in \pi^{\prime}$ the following statements hold: 
\begin{itemize}
\item[(i)] if $S$ contains a player $x$ with $x \in X_t\setminus X_{t^{\prime}}$ 
for every child $t^{\prime}$ of $t$, then $S \subseteq X_t$.
\item [(ii)] $S$ can appear in only one subtree, i.e., 
if $S\cap X_t$ is a subset of both $X_{t_1}$ and $X_{t_2}$ for some children $t_1,t_2$ of $t$, 
then $S \subseteq V_{t_1}$ or $S \subseteq V_{t_2}$.
\end{itemize}
\end{lemma}
\begin{proof}
Consider a node $t \in V(T)$. Note that every optimal coalition structure can be transformed into another 
optimal coalition structure where each winning coalition is minimal. Thus, let $\pi$ be such an optimal 
coalition structure for $V_{t}$. Take an arbitrary coalition  $S \in \pi$. 
If $S$ is losing, then it is clear that $S$ can be divided into $S\cap X_t$ and 
$S \cap (V_{t^{\prime}} \setminus X_t)$ for each child $t^{\prime}$ of $t$ without changing the 
sum $\sum_{S \in \pi}v(S)$, and hence our claims hold. Now, suppose that 
$S$ is a minimal winning coalition. By Theorem \ref{thm:MWC}, $S$ forms a clique.

To show (i), suppose that there is a player $x \in S$ such that $x \in X_t \setminus X_{t^{\prime}}$ for every 
child $t^{\prime}$ of $t$. Assume towards a contradiction that $S \not \subseteq X_t$ and hence 
there is a player $y \in S \setminus X_t$. Observe that no bag $X_w$, $w \in V(T)$, 
contains both $x$ and $y$, since $x$ is not present in 
$X_{t^{\prime}}$ for any successor $t^{\prime}$ of $t$, and $y$ does not appear above $t$. However, since $S$ is a 
clique, $x$ is adjacent to $y$ in the graph $G^{+}_v$, which means that there is a bag containing both $x$ and $y$. 
This contradicts requirement (ii) of Definition~\ref{def:treedecomposition}.

To show (ii), suppose that $S\cap X_t$ is a subset of both $X_{t_1}$ and $X_{t_2}$ for some children $t_1,t_2$ of $t$. 
Assume towards a contradiction that $S \not \subseteq V_{t_1}$ and $S \not \subseteq V_{t_2}$; 
thus there exist a player $x \in S \setminus V_{t_1}$ and a player $y \in S \setminus V_{t_2}$. 
Since the intersection $S\cap X_t$ is fully contained in both $V_{t_1}$ and $V_{t_2}$, 
we have $x, y \not \in X_{t}$, and hence no bag contains both $x$ and $y$. 
However, since $S$ is a clique, $x$ and $y$ are adjacent, contradicting requirement (ii) 
of Definition~\ref{def:treedecomposition}. 
\end{proof}

By Lemma \ref{lem:treedecomposition}, one can find an optimal coalition structure for a simple game by trying all 
possible partitions of each bag and combining them in a bottom-up manner. 
Before we present the proof of Theorem \ref{thm:coalitionstructure}, we need a few auxiliary definitions.

Consider a cooperative game $(N, v)$ and a coalition structure $\pi$ for this game.
For each subset $S \subseteq N$, we define $\pi(S)$ to be the coalition in $\pi$ containing $S$ 
if such a coalition exists, and $\pi(S)=\emptyset$ otherwise. 

A tree decomposition $(T,(X_{t})_{t \in V(T)})$
of a graph $G$ is {\em nice} if for the root $r$ of $T$ we have $X_r=\emptyset$, 
and each node belongs to one of the following types:
\begin{itemize}
\item Leaf: $t$ is a leaf in $T$ and $|X_t|=1$. 
\item Introduce: $t$ has one child $t^{\prime}$, and $X_t =X_{t^{\prime}}\cup \{x\}$ for some $x  \not \in X_{t^{\prime}}$.
\item Forget: $t$ has one child $t^{\prime}$, and $X_t=X_{t^{\prime}} \setminus \{x\}$ for some $x \in X_{t^{\prime}}$.
\item Join: $t$ has two children $t_1,t_2$ such that $X_t=X_{t_1}=X_{t_2}$.
\end{itemize}

We are now ready to prove Theorem \ref{thm:optcoalition}.

\begin{theorem}\label{thm:optcoalition}
Consider a simple game $(N, v)$.
There exists an algorithm that, given oracle access to $v$ and a tree decomposition $(T,(X_{t})_{t \in V(T)})$ 
of the supermodular dependency graph $G^+_v$ with treewidth $\tw$, 
computes an optimal coalition structure in time $O((\tw+1)^{\tw+1}4^{\tw+1}|V(t)|)$.
\end{theorem}
\begin{proof}
First, recall that in ${\tw}^2n$ time we can transform a given tree decomposition of treewidth $\tw$ with $n$ 
nodes into a nice one with the same treewidth $\tw$ and $O(\tw \cdot n)$ nodes \cite{Cygan2015}. 
In what follows, let $(T,(X_{t})_{t \in V(T)})$ denote such a decomposition.

Now, we give a dynamic program over the tree decomposition as follows. 
For each node $t \in V(T)$, each coalition structure $\pi$ of $X_t$, and each subset $\pi_{\out}$ of $\pi$, 
we define $\opt[t,\pi,\pi_{\out}]$ to be the maximum value $\sum_{S \in \pi^{*}}v(S)$ 
such that $\pi^{*}$ is a coalition structure of $V_t$ where all the coalitions in $\pi$ are {\em extended} in 
$\pi^*$ without changing the coalitions in $\pi \setminus \pi_{\out}$, i.e., for all $S \in \pi$, $S \subseteq 
\pi^{*}(S)$, and for all $S \in \pi \setminus \pi_{\out}$, $S=\pi^{*}(S)$. Starting from the leaves and going up to 
the root, we will fill out the dynamic programming table. The case where $t$ is a leaf corresponds to the base case 
of the recurrence; we then compute values for a non-leaf node $t$ based on the values for the children of $t$. 
We will finally obtain $\opt[r,\{\emptyset\},\emptyset]$, which is the value 
we want to compute.

{\em Leaf}: If $t$ is a leaf node, then we have only one value 
$\opt[t,\{X_t\},\{X_t\}]=\opt[t,\{X_t\},\emptyset]=v(X_t)$.

{\em Introduce}: Suppose $t$ is an introduce node with child $t^{\prime}$ such that $X_t=X_{t^{\prime}}\cup \{x\}$. 
Let $S_x$ be a coalition in $\pi$ containing $x$, and let 
$\pi^{\prime}=(\pi\setminus \{S_x\}) \cup \{S_x\setminus \{x\}\}$. 
We claim that $\opt[t,\pi,\pi_{\out}]$ is given by the value of an optimal coalition structure of the 
subtree rooted at the child $t^{\prime}$ and the marginal contribution of $x$ to the coalition $S_x\setminus \{x\}$, 
namely,
\[
\opt[t,\pi,\pi_{\out}]=\opt[t^{\prime},\pi^{\prime},\pi^{\prime}_{\out}]+(v(S_x)-v(S_x\setminus \{x\})),
\]
where $\pi^{\prime}_{\out}=\pi_{\out}\setminus \{S_x\}$.

To see this, let $\pi^{*}$ be a partition of $V_t$ that attains the maximum in the definition of 
$\opt[t,\pi,\pi_{\out}]$ and satisfies condition (i) of Lemma~4. 
Since $x \in X_t\setminus X_{t^{\prime}}$, we have $S_x \subseteq X_t$ and hence $S_x \in \pi^{*}$.
Then, it follows that $\pi^{**}=(\pi^{*} \setminus \{S_x\})\cup \{S_x\setminus \{x\}\}$ 
is a partition considered in the definition of $\opt[t^{\prime},\pi^{\prime},\pi^{\prime}_{\out}]$, 
and we have $\sum_{X \in \pi^{**}}v(X) \leq \opt[t^{\prime},\pi^{\prime},\pi^{\prime}_{\out}]$. Hence, 
\begin{align*}
&\opt[t,\pi,\pi_{\out}]= \sum_{X \in \pi^{**}}v(X)+(v(S_x)-v(S_x\setminus \{x\}))\\
&\leq \opt[t^{\prime},\pi^{\prime},\pi^{\prime}_{\out}]+(v(S_x)-v(S_x\setminus \{x\})). 
\end{align*}

Conversely, let $\pi^{*}$ be a partition of $V_{t^{\prime}}$ that attains the maximum in the definition of 
$\opt[t^{\prime},\pi^{\prime},\pi^{\prime}_{\out}]$. Then, since $S_x\setminus \{x\} \in \pi^{\prime}$, but 
$S_x\setminus \{x\} \not \in \pi^{\prime}_{\out}$, the coalition $S_x\setminus \{x\}$ remains the same in $\pi^*$, 
i.e., $S_x\setminus \{x\} \in \pi^{*}$; thus, 
partition $\pi^{**}=(\pi^{*}\setminus \{S_x\setminus \{x\}\}) \cup \{S_x\}$ 
is considered in the definition of $\opt[t,\pi,\pi_{\out}]$, and
\begin{align*}
&\opt[t,\pi,\pi_{\out}] \geq \sum_{X \in \pi^{**}}v(X),\\
&= \sum_{X \in \pi^{*}}v(X)+(v(S_x)-v(S_x\setminus \{x\})),\\
&= \opt[t^{\prime},\pi^{\prime},\pi^{\prime}_{\out}]+(v(S_x)-v(S_x\setminus \{x\})).
\end{align*}

{\em Forget}: Suppose $t$ is a forget node with child $t^{\prime}$ such that $\{x\}=X_{t^{\prime}}\setminus X_t$. 
Then, it can be verified that $\opt[t,\pi,\pi_{\out}]$ is given by
\[
\max\{\, \opt[t^{\prime},\pi^{S},\pi^S_{\out}] \mid S \in \pi_{\out}\cup\{\emptyset\} \,\},
\]
where $\pi^{S}=(\pi\setminus \{S\}) \cup \{S\cup \{x\}\}$ and $\pi^S_{\out}=(\pi_{\out}\setminus \{S\}) \cup \{S\cup \{x\}\}$ for $S \in \pi_{\out}\cup \{\emptyset\}$. 

{\em Join}: Finally, suppose that $t$ is a join node with children $t_1,t_2$ such that $X_t=X_{t_1}=X_{t_2}$. Then,
$\opt[t,\pi,\pi_{\out}]$ is the maximum value of 
\[
\opt[t_1,\pi,\pi^{1}_{\out}] + \opt[t_2,\pi,\pi^{2}_{\out}]-\sum_{S \in \pi}v(S)
\]
over all the pairs $(\pi^{1}_{\out},\pi^{2}_{\out})$ where $\pi_{\out}$ is a disjoint union of $\pi^{1}_{\out}$ and $\pi^{2}_{\out}$; intuitively, each $\pi^{i}_{\out}$ specifies how a coalition in $\pi_{\out}$ will be extended to subtrees $V_{t_1}$ and $V_{t_2}$.

To show the correctness of our algorithm, 
let $\pi^{*}$ be a partition of $V_t$ that attains the maximum in the definition of $\opt[t,\pi,\pi_{\out}]$ 
and satisfies condition (ii) of Lemma~4.

Let $\pi_{\mathrm{in}}$ 
be the set of coalitions in $\pi$ that remain the same, 
i.e., $\pi_{\mathrm{in}}=\pi\setminus \pi_{\out}$. 
Observe that by condition (ii) of Lemma~4, each coalition $S \in \pi_{\out}$ 
either remains the same in $\pi^*$ or is extended to a subtree rooted at $t_1$ or $t_2$, 
that is, $\pi^*(S) \subseteq X_t$ or $\pi^*(S) \subseteq V_{t_{i}}$ for some $i=1,2$. 
We now divide $\pi_{\out}$ into two families $\pi^{1}_{\out}$ and $\pi^{2}_{\out}$ depending on how $S \in \pi_{\out}$ has been extended in $\pi^*$; specifically, we let $\pi^{1}_{\out}$ and $\pi^{2}_{\out}$ be subsets of $\pi_{\out}$ such that $\pi_{\out}$ can be represented as a disjoint union of $\pi^{1}_{\out}$ and $\pi^{2}_{\out}$, and
\begin{itemize}
\item $\pi^{1}_{\out}$ includes all the coalitions $S$ in $\pi_{\out}$ such that 
$\pi^{*}$(S) is contained in $V_{t_1}$ only, i.e., $\pi^*(S) \subseteq V_{t_{1}}$ and $\pi^*(S) \not \subseteq V_{t_2}$; and
\item $\pi^{2}_{\out}$ includes all the coalitions $S$ in $\pi_{\out}$ such that 
$\pi^{*}$(S) is contained in $V_{t_2}$ only, i.e., $\pi^*(S) \subseteq V_{t_{2}}$ and $\pi^*(S) \not \subseteq V_{t_1}$.
\end{itemize}
Let $\pi^{1}_{\mathrm{in}} = \pi_{\mathrm{in}} \cup \pi^{2}_{\out}$, 
    $\pi^{2}_{\mathrm{in}} = \pi_{\mathrm{in}} \cup \pi^{1}_{\out}$. For $i=1, 2$, set
\[
\pi^{i}=\pi^{i}_{\mathrm{in}} \cup \{\, \pi^*(S) \mid S \in \pi^{i}_{\out} \,\} 
\cup \{\, S \in \pi^* \mid S \subseteq V_{t_i} \setminus X_t \,\}.
\]
Note that each $\pi^{i}$ consists of $\pi_{\mathrm{in}}$, the coalitions $\pi^*(S) \cap X_{t}$ with $S \in \pi_{\out}$ being extended to another subtree, the coalitions in $\pi^*$ that traverse both $X_{t}$ and $V_{t_i} \setminus X_t$, and the coalitions in $\pi^*$ that are fully contained in $V_{t_i} \setminus X_t$; hence, it can be easily verified that $\pi^{i}$ is a partition of $V_{t_i}$. Now, the optimal value $\opt[t,\pi,\pi_{\out}]$ can be written as follows:
\begin{align*}
&\opt[t,\pi,\pi_{\out}]\\
&= \sum_{S \in \pi_{\mathrm{in}}}v(S)+\sum_{S \in \pi^*\setminus \pi_{\mathrm{in}} }v(S)\\
&=2\sum_{S \in \pi_{\mathrm{in}}}v(S)+\sum_{S \in \pi_{\out}}v(S) + \sum_{S \in  \pi^*\setminus \pi_{\mathrm{in}}}v(S) - \sum_{S \in \pi}v(S)\\
&=\sum_{S \in \pi^1_{\mathrm{in}}}v(S) + \sum_{S \in \pi^2_{\mathrm{in}}}v(S)+ \sum_{S \in  \pi^*\setminus \pi_{\mathrm{in}}}v(S) - \sum_{S \in \pi}v(S)\\
&=\sum_{S \in \pi^1}v(S) + \sum_{S \in \pi^2}v(S) - \sum_{S \in \pi}v(S)\\
&\leq \opt[t_1,\pi,\pi^{1}_{\out}] + \opt[t_2,\pi,\pi^{2}_{\out}]-\sum_{S \in \pi}v(S).
 \end{align*}
Conversely, let $\pi^{1}_{\out}$ and $\pi^{2}_{\out}$ be subsets of $\pi_{\out}$ such that $\pi_{\out}$ can be 
represented as a disjoint union of $\pi^{1}_{\out}$ and $\pi^{2}_{\out}$. Suppose that for each $i=1,2$, $\pi^{i}$ 
is a partition of $V_{t_i}$ that attains the maximum in the definition of $\opt[t_i,\pi,\pi^i_{\out}]$.
Let $\pi_{\mathrm{in}}$ be the set of coalitions in $\pi$ that should be preserved, 
i.e., $\pi_{\mathrm{in}}=\pi\setminus \pi_{\out}$, 
and let $\pi^{i}_{\mathrm{in}}$ be the set of coalitions in $\pi$ that remain the same in $\pi^i$, 
i.e., $\pi^{i}_{\mathrm{in}}=\pi \setminus\pi^{i}_{\out}$ for each $i=1,2$. We define
\[
\pi^*=\pi_{\mathrm{in}} \cup (\pi^1 \setminus \pi^1_{\mathrm{in}}) \cup (\pi^2 \setminus \pi^2_{\mathrm{in}}).  
\]
That is, $\pi^*$ consists of $\pi_{\mathrm{in}}$, the coalitions in $\pi^i$ that have been extended in each 
$V_{t_i}$, and the coalitions in $\pi^i$ that are fully contained in $V_{t_i} \setminus X_t$; hence, it can be 
verified that $\pi^*$ is a partition of $V_t$. By a similar calculation as above,
\begin{align*}
&\opt[t_1,\pi,\pi^{1}_{\out}] + \opt[t_2,\pi,\pi^{2}_{\out}]-\sum_{S \in \pi}v(S)\\
&=\sum_{S \in \pi^1}v(S) + \sum_{S \in \pi^2}v(S) - \sum_{S \in \pi}v(S)\\
&=\sum_{S \in \pi^1_{\mathrm{in}}}v(S) + \sum_{S \in \pi^2_{\mathrm{in}}}v(S)+ 
\sum_{S \in \pi^*\setminus \pi_{\mathrm{in}}}v(S) - \sum_{S \in \pi}v(S)\\
&=2\sum_{S \in \pi_{\mathrm{in}}}v(S)+\sum_{S \in \pi_{\out}}v(S) + \sum_{S \in \pi^*\setminus \pi_{\mathrm{in}}}v(S) 
- \sum_{S \in \pi}v(S)\\
&= \sum_{S \in \pi_{\mathrm{in}}}v(S)+\sum_{S \in \pi^*\setminus \pi_{\mathrm{in}}}v(S)\\
&=  \sum_{S \in \pi^*}v(S) \leq \opt[t,\pi,\pi_{\out}].
 \end{align*}
The overall running time of our algorithm is 
$$
O((\tw+1)^{\tw+1}4^{\tw+1}|V(t)|),
$$ 
since the size of the dynamic programming table is $O((\tw+1)^{\tw+1}2^{\tw+1}|V(t)|)$ 
and each entry can be filled in $O(2^{\tw+1})$ time. This completes the proof.
\end{proof}

\citet{Voice2012} also studied the problem of finding an optimal coalition structure for a graph-restricted 
instance, and designed a $O(\tw^{\tw+O(1)}n)$ algorithm for a graph with treewidth $\tw$ when players disconnected on 
a graph do not depend on each other. In our setting, the result of \citet{Voice2012} translates into an efficient algorithm for 
finding an optimal coalition structure for games whose dependency graph has bounded treewidth. However, this 
result does not imply ours, since Theorem \ref{thm:optcoalition} is for games whose supermodular dependency graph 
has bounded treewidth, and the supermodular dependency graph is a subgraph of the dependency graph.

Recall that a chordal graph always admits a clique-tree decomposition, i.e., a tree decomposition where each bag 
forms a maximal clique; moreover, such a decomposition can be found in linear time \cite{Berry2017}. Hence 
the treewidth of a chordal graph is bounded by the maximum degree of the graph, 
and we have the following corollary.

\begin{corollary}\label{FPT:optcoalition}
There exists an algorithm that, given a weighted voting game $[N;\bfw;q]$ and its supermodular dependency graph, 
computes an optimal coalition structure in time $O((p+1)^{p+1}4^{p+1}(p+1)n)$. 
\end{corollary}

Our FPT result with respect to the supermodular degree does not extend to arbitrary simple games;
indeed, for such games we obtain a hardness result that holds even if the supermodular degree 
is bounded by a constant (see the supplemental material for the proof).

\begin{theorem}\label{NPh:optcoalition}
Finding an optimal coalition structure of a simple game is {\em NP}-hard even if the supermodular degree is at most $6$.
\end{theorem}
\begin{proof}
We reduce from the NP-complete problem {\sc Exact-3-Cover} (X3C) \cite{garey}. 
Given a set of elements $X = \{x_1,x_2,...,x_{3n}\}$ and a family 
$\calS = \{S_1,S_2,...,S_m\}$ of three-element subsets of $X$, this problem 
asks whether $X$ can be covered by $n$ sets from $\calS$. This problem remains NP-complete if for each element 
$x \in X$ its frequency $p_x =|\{\, S \in \calS \mid x \in S\,\}|$ is at most $3$.

Given a set of elements $X = \{x_1,x_2,...,x_{3n}\}$ and a family $\calS = \{S_1,S_2,...,S_m\}$ of three-element 
subsets of $X$ where $p_x\le 3$ for each $x \in X$, we construct a simple game $G=(N, v)$ 
where $N=X$, and for each subset $S \subseteq N$, we set $v(S)=1$ if there exists 
a set $S_j\in \calS$ such that $S_j \subseteq S$, and $v(S)=0$ otherwise. 
Notice that every player belongs to at most three minimal winning coalitions and hence positively 
depends on at most six players by Theorem~3. It is immediate that an optimal coalition 
structure has value $n$ if and only if $X$ can be covered by $n$ sets from $\calS$.
\end{proof}


\section{Constructing Dependency Graphs}
In this section, we explore another question: if the dependency graph is not given to us as input, 
how hard is it to compute the dependency/supermodular degree? 

An interesting class of games where constructing dependency graphs is easy is {\em induced subgraph games}.
Formally, an induced subgraph game with a set of players $N$ is given by an undirected graph $G=(N,E)$ 
and a weight function $w:E \rightarrow \bbR$. For each $S\subseteq N$ 
the value $v(S)$ of a coalition $S$ is given by the 
sum of edge weights in the graph induced by $S$. In such games, player $i$ 
depends on player $j$ if and only if the edge $(i, j)$ has non-zero weight; 
furthermore, their dependence is positive if and only if the weight is positive.

However, for general games constructing the dependency graph is not easy.
As we have seen before, a dummy player corresponds to an isolated vertex of 
both dependency and supermodular dependency graphs; deciding whether a player is a dummy is NP-hard, 
for instance, in weighted voting games \cite{Chalkiadakis2011} or 
in threshold network flow games \cite{Bachrach2008}. This immediately implies 
the hardness of constructing a supermodular dependency graph or a dependency graph for
these classes of games. 
On the positive side, for weighted voting games we can determine if there is a dependency
between players in pseudo-polynomial time.

\begin{theorem}
Given a weighted voting game $[N;\bfw;q]$ with integer weights
and players $i,j \in N$ with $i \neq j$, one can decide whether $i$ 
depends on $j$ in time polynomial in $n$ and $w_{\max}$, 
where $w_{\max}$ is the maximum weight among the players.
\end{theorem}
\begin{proof}
We need to check if there is a coalition $S \subseteq N \setminus \{i,j\}$ 
such that 
(a) coalition $S\cup \{i,j\}$ is winning, 
but coalitions $S\cup \{i\}$ and $S\cup \{j\}$ are losing, or 
(b) coalitions $S\cup \{i\}$ and $S\cup \{j\}$ are winning, but $S$ is losing. 
Note that condition (a) is equivalent to $q-(w_i+w_j) \leq w(S) < \min \{q-w_i,q-w_j \}$, 
and condition (b) is equivalent to $\max \{q-w_i,q-w_j \} \leq w(S) <q$.
Both conditions can be checked by considering at most $w_{\max}$ weights, and, for each
weight, checking if $N\setminus\{i, j\}$ contains a coalition of that weight;
the latter question is effectively an instance of {\sc Subset Sum}
and hence can be answered in pseudo-polynomial time.
\end{proof}

One can also ask if it is possible to construct the dependency graph 
by evaluating the characteristic function at polynomially many points.
However, perhaps unsurprisingly, the query complexity of the dependency graph
turns out to be exponential in the number of players.

\begin{theorem}
Any algorithm that computes the dependency graph of a simple game must evaluate 
the characteristic function in at least  $\binom{n}{n/2}$ points in the worst case.
\end{theorem}
\begin{proof}
We will describe a class of simple games where $\binom{n}{n/2}$ value queries are required 
to decide whether a player is a dummy.

Let $n$ be an even number, and set $X=[n]$. We introduce one new player $a \not \in X$.
For each subset $H\subseteq [n]$ with $|H|=n/2$, we construct a simple game $\calG_H=(N, v_H)$.
The player set in this game is $N=X\cup \{a\}$, and $v_H(S)=1$ if and only if 
$|S\cap X|\ge \frac{n}{2}$ and $S\neq H$ or if $S=H\cup\{a\}$.
Notice that player $a$ is not a pivot for any coalition in $\calG_H$ except for $H$.
We also define a game $\calG_0=(N, v_0)$, where $N=X\cup\{a\}$ and $v_0(S)=1$
if and only if $|S\cap X|\ge \frac{n}{2}$. Clearly, $a$ is a dummy in $\calG_0$.

Consider any algorithm that constructs a dependency graph using value queries. 
If it asks fewer than $\binom{n}{n/2}$ queries, then there is a coalition
$H\subseteq X$ with $|H|=n/2$ whose value has not been queried.
Now suppose that whenever the algorithm queries the value of a coalition
of size $n/2$, the answer is $1$. Then at the end the algorithm is unable to distinguish
between $\calG_H$ and $\calG_0$, and hence it is unable to decide whether $a$ is a dummy.
\end{proof}

\section{Conclusions}
In this work, we have demonstrated that the concepts of dependency degree and supermodular degree
are useful in the analysis on cooperative games: we have obtained FPT results with respect 
to these parameters for a number of solution concepts in cooperative game theory. We have also
explored the limitations of this approach, proving hardness results for games where
these parameters are bounded by a constant. In the future, it would be interesting
to extend this line of work to other solution concepts, such as the nucleolus, the kernel, 
or the bargaining set, or to other succinctly representable classes of cooperative games.

\section*{Acknowledgements}
This work was supported by the European Research Council (ERC) under grant number 639945 (ACCORD).

\bibliographystyle{aaai}

\end{document}